\documentclass[conference,letterpaper]{IEEEtran}

\usepackage[utf8]{inputenc} 
\usepackage[T1]{fontenc}
\usepackage{url}
\usepackage{ifthen}
\usepackage[cmex10]{amsmath} 

\usepackage{amssymb,amsthm,amsmath}
\usepackage{mathtools}
\usepackage{float}
\usepackage{array}
\usepackage{tikz}
\usepackage{algpseudocode}
\usepackage{algorithm}
\usepackage{amssymb}
\usepackage{amsthm}

\usetikzlibrary{shapes.geometric}
\newfloat{procedure}{htbp}{loa}
\floatname{procedure}{Procedure} 
\usepackage{enumitem}
\usepackage[square,numbers]{natbib}
\usepackage{soul}
\usepackage{float}
\usepackage{subfig}
\usepackage{caption}
\usepackage{tikz}
\usepackage{capt-of}
\usepackage{xcolor}
\usepackage{changes}
\usepackage{breqn}

\newtheorem{theorem}{Theorem}

\newtheorem{definition}{Definition}

\newtheorem{lemma}{Lemma}

\newtheorem{proposition}{Proposition}
\newtheorem{remark}{Remark}

\newcommand{\R}{\mathbb{R}} 

\newcommand{\Sbold}{\mathbf S}
\newcommand{\x}{\mathbf x}

\newcommand{\calP}{\mathcal{P}}
\newcommand{\mbbS}{\mathbb{S}}

\newcommand{\g}{g} 
\newcommand{\States}{\mathbb{S}}
\newcommand{\cost}{\mathcal{L}}

\newcommand{\lp}{\left(}
\newcommand{\rp}{\right)}

\newcommand{\Cost}{{\mbox{Cost}}}
\newcommand{\OPT}{{\mbox{OPT}}}
\newcommand{\Tie}{{\mbox{Tie}}}
\newcommand{\MCP}{\mathbb{H}}
\newcommand{\LTUN}{{L_{{Tun}}}}

\usepackage{xcolor}
\usepackage{algpseudocode}
\usepackage[parfill]{parskip}
\usepackage{amsmath}
\usepackage{comment}
\usepackage{caption}
 \usepackage{capt-of}
\usetikzlibrary{shapes.geometric}

\DeclareMathOperator*{\argmax}{argmax}

\begin{document}
\title{A (Weakly)  Polynomial Algorithm for AIVF Coding} 


\author{%
  \IEEEauthorblockN{Mordecai J.~Golin}
  \IEEEauthorblockA{University of Massachusetts, Amherst\\
                    Email: mgolin@mass.edu}
  \and
  \IEEEauthorblockN{Reza Hosseini Dolatabadi and Arian Zamani}
  \IEEEauthorblockA{Department of Computer Engineering\\ 
                    Sharif University of Technology, Tehran, Iran\\
                    Emails: \{reza.dolatabadi256\}, \{arian.zamani243\}@sharif.edu}
}


\maketitle


\begin{abstract}
It is possible to improve upon Tunstall coding using a collection of multiple parse trees.  The best such results so far are Iwata and Yamamoto's maximum cost AIVF codes.  The most efficient algorithm for designing such codes
is an iterative one that could run in exponential time.  In this paper, we show that this problem fits into the framework of a newly developed technique that uses linear programming with the Ellipsoid method to solve the minimum cost Markov chain problem. This permits constructing maximum cost AIVF codes in (weakly) polynomial time.
\end{abstract}

\section{Introduction}\label{sec: intro}
Consider a stationary memoryless source with alphabet $S=\{a_0,a_1,...,a_{|S|-1}\}$ such that the symbol $a_i$ is generated with probability $p(a_i)$ for each $i\in \mathcal{I}_{|S|}$, where $\mathcal{I}_{n}$ denotes the integer set $\{0,1,..,n-1\} $. Without loss of generality we assume that $p(a_0)\geq p(a_1)\geq ... \geq p(a_{|S|-1})$. 
Tunstall Codes ~\cite{tunstall-original} 
map a dictionary $\mathcal{D}$ of variable-length prefix-free phrases of the source symbols into codewords of fixed-length.

Tunstall Coding is usually considered optimal amongst variable-to-fixed (VF) codes in the sense that it achieves the smallest coding rate $R(\mathcal{D},s):=\frac{\log D}{E[L(S)]}$, where $E[L(S)]$ is the average length of the phrases parsed by dictionary $\mathcal{D}$ and $D=|\mathcal{D}|$. More accurately, Tunstall Coding is only optimal among prefix-free dictionaries. However, VF codes don't need to have prefix-free Dictionaries. The dictionary only needs to be exhaustive, meaning that if the current source sequence is sufficiently large,  it is always possible to match a prefix of the source sequence with a parseword from the dictionary. Yamamoto and Yokoo ~\cite{aivf-original} proposed  AIVF (almost instantaneous VF) codes; this was later extended by Dubé and Haddad ~\cite{dube-haddad}. This generalization of Tunstall coding uses multiple dictionaries that aren't necessarily prefix-free and as a result, it achieves a better coding rate than Tunstall Coding.
Very recently,  Iwata and  Yamamoto \cite{AIVF-iterative} described an exponential time method for constructing the maximum  cost (minimum coding rate) AIVF codes for a given $D$.

In this paper we show that this problem fits into the framework of a newly developed technique  \cite{golinaifv-m} that uses linear programming with the Ellipsoid method
\cite{ellipsoid} to solve the minimum cost markov chain problem. This will permit constructing maximum cost AIVF codes for a given $D$ in time polynomial  in $|S|$, $D$ and $b$, where $b$ is the maximum number of bits needed to encode any of the $p(a_i).$

{\em Note: This is a purely theoretical result that demonstrates that the problem is technically  polynomial-time solvable. It is not practically implementable because of its use of the Ellipsoid method as a subroutine.
An  obvious remaining open question is to find an efficient implementable polynomial time algorithm.}

The next subsection quickly describes Tunstall and then AIVF coding,
For consistency,  notation and examples are mostly the same as those given in ~\cite{AIVF-iterative}.

\subsection{Tunstall Coding}
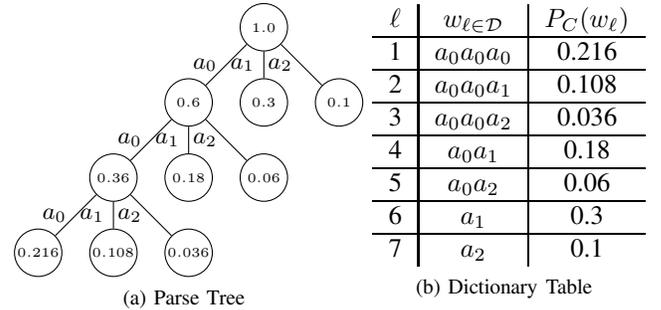
\begin{figure}[ht] 
    \centering
    \subfloat[Parse Tree]{
\begin{tikzpicture}[
node/.style={circle,draw,minimum width=6.3mm,inner sep=0}, level 1/.style={sibling distance=10mm},level distance=10mm, edge from parent path={(\tikzparentnode) -- (\tikzchildnode)},baseline
]
]
\node [node]{\tiny $1.0$}
[child anchor=north]
child {node[node]{\tiny $0.6$}
child {node[node]{\tiny $0.36$} 
child {node[node]{\tiny $0.216$} edge from parent node[left,draw=none,yshift=0mm,xshift=0mm]{\small $a_0$}}
child {node[node]{\tiny $0.108$} edge from parent node[left,draw=none,yshift=0mm,xshift=0mm]{\small $a_1$}}
child {node[node]{\tiny $0.036$} edge from parent node[left,draw=none,yshift=0mm,xshift=0mm]{\small $a_2$}}
edge from parent node[left,draw=none,yshift=0mm,xshift=0mm]{\small $a_0$}}
child {node[node]{\tiny $0.18$} edge from parent node[left,draw=none,yshift=0mm,xshift=0mm]{\small $a_1$}}
child {node[node]{\tiny $0.06$} edge from parent node[left,draw=none,yshift=0mm,xshift=0mm]{\small $a_2$}}
edge from parent node[left,draw=none,yshift=0mm,xshift=0mm]{\small $a_0$}
}
  child {node[node]{\tiny $0.3$} edge from parent node[left,draw=none,yshift=0mm,xshift=0mm]{\small $a_1$}}
  child {node[node]{\tiny $0.1$} edge from parent node[left,draw=none,yshift=0mm,xshift=0mm]{\small $a_2$}};
\end{tikzpicture}}
\subfloat[Dictionary Table]{
\begin{tabular}[t]{c|c|c}
$\ell$ & $w_{ \ell \in\mathcal{D}}$  & $P_{C}(w_{ \ell})$                     \\ \hline
1 & $a_0 a_0 a_0$ & 0.216 \\ \hline
2 & $a_0 a_0 a_1$ & 0.108                     \\ \hline
3 & $a_0 a_0 a_2$ & 0.036                     \\ \hline
4 & $a_0 a_1$  & 0.18                      \\ \hline
5 & $a_0 a_2$  & 0.06                      \\ \hline
6 & $a_1$   & 0.3                       \\ \hline
7 & $a_2$   & 0.1                      
\end{tabular}}
\caption{(Figure 1 in \cite{AIVF-iterative}) The parse tree and dictionary $\mathcal{D}$ of the Tunstall code for source of $S=\{a_0,a_1,a_2\}$ with $p(a_0)=0.6$, $p(a_1)=0.3$, $p(a_2)=0.1$ and $D=7$. }
\label{fig:tunstall example}
\end{figure}

\begin{algorithm}[]
\caption{Tunstall algorithm}\label{alg:tun}
\begin{algorithmic}[1]
\Require $k\in\mathbb{Z}^{+}$ and $\{p(a_i)$ $|$ $i\in \mathcal{I}_{|S|}\}$.
\Ensure $\mathcal{D}$ such that $D=|S|+(|S|-1)k$.
\algrenewcommand\algorithmicrequire{\textbf{Initialization:}}
\Require $t \gets $Root +$S$.
\Comment{Root has only one node}
\For{$i=1$ to $k$}
\State $w_{max} \gets $ $\argmax\limits_{w\in\mathcal{D}(t)} p_{C}(w)$ \Comment{$\mathcal{D}(t)$ is the dictionary induced by parse tree t.}
\State $t \gets t+\{w_{max} \cdot a_i $ $|$ $i\in \mathcal{I}_{|S|}\}$
\EndFor \\
\Return $\mathcal{D}(t)$
\end{algorithmic}
\end{algorithm}

Figure \ref{fig:tunstall example} shows an example of a Tunstall code with its associated 
 parse tree and dictionary. Every internal node has $|S|$ children, each corresponding to a different source symbol with a different codeword assigned to each leaf. The parse dictionary $\mathcal{D}=\{w_\ell|\ell=1,2,...,D\}$ consists of the source symbol strings from the root to all the leaves, and $D=|S|+(|S|-1)k$ for some $k\geq0$.

In order to encode the source sequence using dictionary $\mathcal{D}$ we first map a prefix $w_\ell \in \mathcal{D}$ of the source sequence to the corresponding code word $\ell$ of fixed-length and then remove the prefix $w_\ell$ from the source sequence. We continue this encoding process until the source sequence becomes empty (ignoring the case in which a few symbols might be left that cannot be parsed by the dictionary).

Since  the symbols of the source sequence are generated independently we know that the probability $p_{W}(w)$ of phrase $w\in S^{+}$ 
satisfies
$p_{W}(w)=\prod_{i=1}^{n}p(c_i),$
where $w=c_1c_2...c_n$,
$c_i \in S,$
and $p_{W}(\lambda)=1$ where $\lambda$ is the empty string.  Figure \ref{fig:tunstall example} illustrates that 
that the value of a node is the probability $p_{W}(w)$ of the phrase $w$ read by the path from the root to that node. In the case of Tunstall codes the occurrence probability $p_C(w_\ell)$ of a parse phrase $w_\ell \in \mathcal{D}$ is given by
\begin{equation} \label{to:tunstall occurrence}
    p_C(w_\ell)=p_{W}(w_\ell).
\end{equation}
The average parse length 
$E[\LTUN]$
of Tunstall code with dictionary $\mathcal{D}$ for source $S$ is therefore  
\begin{equation}\label{tunstall average}
    E[\LTUN(S)]:=\sum_{w_\ell \in \mathcal{D}}p_C(w_\ell)l(w_\ell),
\end{equation}
where $l(w)$ denotes the length of $w$. For the example shown in Figure \ref{fig:tunstall example},  $E[\LTUN(S)]=1.96$.

The Tunstall coding algorithm is given in 
Algorithm \ref{alg:tun}. For given source $S$ and $k >0$ It is known to construct the    dictionary $\mathcal{D}$ satisfying  $D=|S|+(|S|-1)k$, and maximizing $E[\LTUN(S)]$. 
For a tree $t$ and set of strings $\{w\}$, $t+\{w\}$ means an extension of $t$ by adding edges and nodes corresponding to the root-to-leaf paths given by strings $\{w\}$. The algorithm first starts by taking a tree with one internal node and each symbol as its children. Then for $k$ iterations, it finds the leaf node with the highest corresponding $p_{C}(w)$ and then expands that node by giving it $|S|$ children, each having an edge corresponding to a different symbol.

\subsection{AIVF Coding}
\begin{figure}[ht] 
    \centering
    \subfloat[The parse tree $t_0$]{
\begin{tikzpicture}[
node/.style={circle,draw,minimum width=7.3mm,inner sep=0}, level 1/.style={sibling distance=10mm},level distance=10mm, edge from parent path={(\tikzparentnode) -- (\tikzchildnode)},baseline
]

\node [node]{\tiny $1.0$}
[child anchor=north]
child {node[node]{\tiny $0.6$}
child {node[node]{\tiny $0.36$} 
child {node[node]{\tiny $0.216$} 
child {node[node]{\tiny $0.1296$} edge from parent node[left,draw=none,yshift=0mm,xshift=0mm]{\small $a_0$}}
edge from parent node[left,draw=none,yshift=0mm,xshift=0mm]{\small $a_0$}}
edge from parent node[left,draw=none,yshift=0mm,xshift=0mm]{\small $a_0$}}
edge from parent node[left,draw=none,yshift=0mm,xshift=0mm]{\small $a_0$}
}
  child {node[node]{\tiny $0.3$} 
  child {node[node]{\tiny $0.18$} edge from parent node[left,draw=none,yshift=0mm,xshift=0mm]{\small $a_0$}}
  edge from parent node[left,draw=none,yshift=0mm,xshift=0mm]{\small $a_1$}}
  child {node[node]{\tiny $0.1$} edge from parent node[left,draw=none,yshift=0mm,xshift=0mm]{\small $a_2$}};
\end{tikzpicture}}
\subfloat[Dictionary Table of $\mathcal{D}(t_0)$]{
\begin{tabular}[t]{c|c|c}
$\ell$ & $w_{ \ell \in\mathcal{D}(t_0)}$  & $P_{C_{0}}(w_{ \ell})$                     \\ \hline
1 & $a_0$ & 0.24 \\ \hline
2 & $a_0 a_0$ & 0.144                     \\ \hline
3 & $a_0 a_0 a_0$ & 0.0864                     \\ \hline
4 & $a_0 a_0 a_0 a_0$  & 0.1296                      \\ \hline
5 & $a_1$  & 0.12                      \\ \hline
6 & $a_1 a_0$   & 0.18                       \\ \hline
7 & $a_2$   & 0.1                      
\end{tabular}}

    \subfloat[The parse tree $t_1$]{
\begin{tikzpicture}[
node/.style={circle,draw,minimum width=7.3mm,inner sep=0},noderight/.style={circle,draw,minimum width=7.3mm,inner sep=0, xshift={10mm}}, level 1/.style={sibling distance=10mm},level distance=10mm, edge from parent path={(\tikzparentnode) -- (\tikzchildnode)},baseline
]
\node [node]{\tiny $1.0$}
[child anchor=north]
child {node[node]{\tiny $0.75$}
child {node[node]{\tiny $0.45$} 
child {node[node]{\tiny $0.27$} 
child {node[node]{\tiny $0.162$} edge from parent node[left,draw=none,yshift=0mm,xshift=0mm]{\small $a_0$}}
edge from parent node[left,draw=none,yshift=0mm,xshift=0mm]{\small $a_0$}}
edge from parent node[left,draw=none,yshift=0mm,xshift=0mm]{\small $a_0$}}
child {node[node]{\tiny $0.225$} edge from parent node[left,draw=none,yshift=0mm,xshift=0mm]{\small $a_1$}}
child {node[node]{\tiny $0.075$} edge from parent node[left,draw=none,yshift=0mm,xshift=0mm]{\small $a_2$}}
edge from parent node[left,draw=none,yshift=0mm,xshift=0mm]{\small $a_1$}
}
  child {node[node]{\tiny $0.25$} 
  child {node[noderight]{\tiny $0.15$} edge from parent node[right,draw=none,yshift=0mm,xshift=0mm]{\small $a_0$}}
  edge from parent node[left,draw=none,yshift=0mm,xshift=0mm]{\small $a_2$}};
\end{tikzpicture}}
\subfloat[Dictionary Table of $\mathcal{D}(t_1)$]{
\begin{tabular}[t]{c|c|c}
$\ell$ & $w_{ \ell \in\mathcal{D}(t_1)}$  & $P_{C_{1}}(w_{ \ell})$                     \\ \hline
1 & $a_1 a_0$ & 0.18 \\ \hline
2 & $a_1 a_0 a_0$ & 0.108                     \\ \hline
3 & $a_1 a_0 a_0 a_0$ & 0.162                     \\ \hline
4 & $a_1 a_1$  & 0.225                      \\ \hline
5 & $a_1 a_2$  & 0.075                      \\ \hline
6 & $a_2$   & 0.1                       \\ \hline
7 & $a_2 a_0$   & 0.15                      
\end{tabular}}
\caption{(Figure 2 in \cite{AIVF-iterative}) The optimal AIVF code for the source  $S=\{a_0,a_1,a_2\}$ with  $p(a_0)=0.6$, $p(a_1)=0.3$, $p(a_2)=0.1$ and $D=7$.}
\label{fig:aivf example}
\end{figure}

The core difference between AIVF Coding and Tunstall coding is the usage of multiple dictionaries instead of one. The AIVF Code proposed by Yamamoto and Yokoo ~\cite{aivf-original} and extended by Dubé and Haddad ~\cite{dube-haddad} uses $|S|-1$ Dictionaries. Let 
${\bf t} = \{t_0,t_1,\ldots t_{|S|-2}\}$ 
denote the set of parse trees of these dictionaries and let $\mathcal{D}(t_i)$ denote the dictionary corresponding to the $t_i$ tree. We now explain the properties of the parse trees following the description in ~\cite{AIVF-iterative}. We will need this description later in the design of the dynamic program.\\
\newline
We assume that $ S= \{a_0,\ldots ,a_{|S|-1}\}$ where $p(a_u) \geq p(a_{u^{\prime}})$ for $0 \leq u < u^{\prime} \leq |S| - 1$, and let $v_j$, $j \in \mathcal{I}_{|S|+1}$, represent a node with $j$ children. Then, the multiple parse trees of an AIVF code satisfy the following properties.
\begin{enumerate}
    \item[(p1)] If $v_j$ is the root of $t_i$, then $1\leq j\leq|S|-i$ and a different source symbol $a_u$, $u \in \{i,...,i+j-1\}$ is assigned to each of edges between the root and its children.
    \item[(p2)] If $v_j$ is a leaf, then $j=0$.
    \item[(p3)] If $v_j$ is an internal node, then $1 \leq j \leq |S|-2$ or $j=|S|$,
and a different source symbol $a_u$, $u\in\mathcal{I}_j$ is assigned to each of edges between node $v_j$ and its children.
\end{enumerate}

The dictionaries $\{\mathcal{D}(t_i)$ $|$ $i\in\mathcal{I}_{|S|-1}\}$ satisfy the exhaustive property which means that if $\mathcal{D}(t_i)$ is our current dictionary and our current source sequence is sufficiently large, then it is possible to match a prefix of the source sequence with a parseword from $\mathcal{D}(t_i)$. 
In tree $t_i$, if a node $v_j$ is not the root and $0\leq j \leq |S|-2$, i.e. $v_j$ is not a complete internal node, then we assign a codeword to the node $v_j$. This means that we assign the codeword to the string $w\in S^{*}$ obtained by the path from the root to the node $v_j$, and we include $w$ in dictionary $\mathcal{D}(t_i)$. Let $\mathcal{D}_{i,j} \subset \mathcal{D}(t_i)$ be the set of parse strings obtained by the paths from the root to all $v_j$ for a fixed $j$ in tree $t_i$.
If $v_j$ is the root of $t_i$ and $j<|S|-i$, we assign a codeword to the root, i.e. we assign the codeword to the empty string $\lambda$, and we add $\lambda$ to $\mathcal{D}_{i,i+j}$. Then, $\mathcal{D}(t_i)$ can be represented as $\mathcal{D}(t_i):= \prod_{j\in \mathcal{I}_{|S|-1}} \mathcal{D}_{i,j}$, $\mathcal{D}_{i,j} \cap \mathcal{D}_{i,j^{\prime}} = \emptyset$ for $j \neq j^{\prime}$.
Note that the root of $t_i$ has $|S|-i$ children at most, and hence the root can have $|S|$ children only in $t_0$.
A source sequence $s=s_1 s_2 ... s_n \in S^{n}$ is encoded by AIVF code with $\{\mathcal{D}(t_i)$ $|$ $i\in\mathcal{I}_{|S|-1}\}$ as follows.
\begin{enumerate}
    \item[(e1)] Use $t_0$ as the current tree for the first parsing of $s$.
    \item[(e2)] Obtain the longest prefix $w$ of $s$ in the current dictionary $\mathcal{D}(t_i)$ by tracing $s$ as long as possible from the root in the current tree $t_i$. Encode the obtained prefix $w$.
    \item[(e3)] If $w$ is included in $\mathcal{D}_{i,j} \subset \mathcal{D}(t_i)$, then use $t_j$ as the current tree for the next parsing. Remove $w$ from the prefix of $s$, and go to $(e2)$ if $s$ is not empty.
\end{enumerate}
Figure \ref{fig:aivf example} shows an example of a  two tree  AIVF code $\{t_0,t_1\}$ using the same parameters as Figure \ref{fig:tunstall example}. This encodes a source sequence ‘${a_1 a_0 a_0 a_0 a_1 a_0 a_0 a_2 a_0 a_2}$’ using the trees in the order $t_0\rightarrow t_0\rightarrow t_1\rightarrow t_1\rightarrow t_0$ and deriving ‘$a_1 a_0$’, ‘$a_0 a_0$’, ‘$a_1 a_0 a_0$’, ‘$a_2 a_0$’, ‘$a_2$’ as the corresponding parsewords.
\newline Now consider the probability $p_{W_i}(w)$ of phrase $w$ from the root to a node in tree $t_i$. From (p1)–(p3) and (e1)–(e3),  when $t_i$ is used to encode $s$ in (e2), the first symbol of $s$ is one of $a_u$, $u = i,...,|S|-1$. Therefore $p_{W_i}(w)$ can be calculated as follows.
\begin{equation} \label{eq:pwi}
     p_{W_i}(w)=
\begin{cases} 
      1 & $if $ w=\lambda, \\
      \frac{p_{W}(w)}{\sum_{u=i}^{|S|-1}p(a_u)} & $otherwise$
   \end{cases}.
\end{equation}
Now  consider the occurrence probability $p_{C_i}(w)$ for a phrase $w\in \mathcal{D}(t_i)$. The codewords are assigned to all leaves and all incomplete nodes in each $t_i$, and $w$ is used by the longest parse rule in the encoding. Then, $p_{C_i}(w)$ is given by
\begin{equation} \label{eq:pci}
    p_{C_i}(w)=p_{W_i}(w)-\sum_{\substack{a\in S s.t. \\ wa \preceq w^{\prime}\in \mathcal{D}(t_i)}} p_{W_i}(wa)
\end{equation}
where $wa\preceq w^{\prime} \in \mathcal{D}(t_i)$ means that $wa$ is a prefix of $w^{\prime}$
and $w^{\prime}\in \mathcal{D}(t_i)$. Then, the average parse length of parse tree $t_i$ , $i \in \mathcal{I}_{|S|-1}$ for source $S$ is given by
\begin{equation}\label{eq:apl}
    E[L_{t_i}(S)]=\sum_{w\in\mathcal{D}(t_i)}p_{C_i}(w)\cdot l(w).
\end{equation}
An  important observation that will be used often is that the height of $t_i$ is at most $D$. This then implies that $E[L_{t_i}(S)] \leq D$ because every parse length in $t_i$ is at most $D$ and therefore the average parse length is also at most $D$.
\newline
In Figure \ref{fig:aivf example} the value in each node is  $p_{W_i}(w)$ of the phrase $w$ derived from the path going from the root to the node; the value $p_{C_i}(w)$ for each $w\in \mathcal{D}(t_i)$ is written in the tables.

Let  $Q$ be the  transition matrix of size $(|S|-1)\cdot(|S|-1)$, whose $(i, j)$-th element is the  conditional probability $q_j(t_i)$, i.e., 
 that $t_j$ is used just after $t_i$ is used. If the prefix $w$ of a source sequence is included in $w \in D_{i,j}$ in the encoding by $t_i$, then $t_j$ is used for the next parsing. Hence, $q_j(t_i)$ is given by
\begin{equation}
    q_j(t_i)=\sum_{w\in D_{i,j}} p_{C_i}(w) \text{ for }  i,j \in \mathcal{I}_{|S|-1}.
\end{equation}

As a transition matrix, $Q$ has a stationary distribution 
$\boldsymbol{\pi}:=(\pi(t_0), \pi(t_1),..., \pi(t_{|S|-2}))$ satisfying $\boldsymbol{\pi}Q=\boldsymbol{\pi}.$

The global average parse length of an AIVF code with $\bold{t}=\{t_i | i \in \mathcal{I}_{|S|-1}\}$ for source S is then given by
\begin{equation}
    E[L_{AIVF}(\bold{t},S)]:=\sum_{i\in \mathcal{I}_{|S|-1}}\pi(t_i)E[L_{t_i}(S)]
\end{equation}
As an example, for the code in Figure \ref{fig:aivf example}
$(E[L_{t_0}],E[L_{t_1}])=(1.8856,2.332)$ and $E[L_{AIFV}(t,S)]=2.10479$ for the example 
which has a better coding rate than the corresponding Tunstall Code illustrated in Figure \ref{fig:tunstall example}.

Let  $\mathcal{T}^{(D)}$ be the set of all possible AIVF codes for $S$ in which each tree has exactly $D$ codewords. The AIVF construction problem is to find a code ${\bf t} \in \mathcal{T}^{(D)}$ that maximizes $E[L_{AIVF}(\bold{t},S)].$

\section{Dynamic Programming  For AIVF} \label{sec: dp for aivf}

Set 
 $\mathcal{T}_{i}$ to be  the set of trees $t_i$ satisfying (p1)-(p3), and $\mathcal{T}_{i}^{(D)}$  the subset of trees in $\mathcal{T}_{i}$ that have exactly $D$ codewords. 
 Let $\x=(x_1,\ldots,x_{|S|-2}) \in \R^{|S|-2}$ and, $\forall t_i\in\mathcal{T}_{i}$,
 define
\begin{equation}
    \Cost\left(t_i,\x\right)=E[L_{t_i}(S)] + \sum_{j=1}^{|S|-2}q_j(t_i)\cdot x_j
\end{equation}
The previous algorithms for finding optimal AIVF codes, required, as subroutines, 
fixing $\x$ and $i\ge 0$ and solving  the local optimization problem of finding the tree $t_i\in\mathcal{T}_{i}^{(D)}$,   that for $i =0$ maximizes 
$\Cost\left(t_i,\x\right)$\
and for $i >0,$ maximizes $\Cost\left(t_i,\x\right) - x_i$. In both cases this is equivalent to  finding the tree that maximizes $\Cost\left(t_i,\x\right)$. Our later algorithm will require solving a slightly generalized version of this problem.

The dynamic programming algorithm we use to solve it is essentially the 
algorithm proposed by Iwata and Yamamoto in ~\cite{AIVF-iterative} which itself was a generalization an extension of the algorithm proposed by Dubé and Haddad ~\cite{dube-haddad} (who only presented/needed it for the case $\x=(0,0,\ldots,0)$)
 rewritten using a  change of variables. Since Iwata and Yamamoto do not give the specific details of their extension and we need a further generalization of that, we provide the full algorithm in the appendix.


\begin{proposition}[\cite{AIVF-iterative}] \label{prop: DP running time}
For any fixed  $i$ and $\x$, dynamic programming can find $t_i\in\mathcal{T}_{i}^{(D)}$ maximizing $\Cost(t_i,\x)$  
 in $O(D^2\cdot |S|)$  time.
\end{proposition}
\textbf{The Generalized Local Optimization  Problem:} \newline
Fix  $P\subseteq \mathcal{I}_{|S|-1}$ such that $0\in P$. Now consider the set 
$\mathcal{T}_{i|P}^{(D)}$ of all trees $t_i\in\mathcal{T}_{i}^{(D)}$  with the further restriction that
$q_j(t_{i})=0$ for every $j\notin P$, i.e., they have transition probability zero 
to  trees with a type not in $P$.

The {\em generalized local optimization problem} is to find $t_i\in\mathcal{T}_{i|P}^{(D)}$ maximizing $\Cost(t_i,\x).$  The reason for needing this will be described later.
We just note here that the dynamic program from Proposition \ref{prop: DP running time} can be easily modified (explained in the appendix) to solve this generalized local optimization problem in $O(D^2\cdot |S|)$  time as well.


\section{The Minimum Cost Markov Chain Problem}
\label{sec: markov chain problem definition}
Define $[m]=\{0,\ldots,m-1 \}.$
Let $\Sbold =  \lp S_0, S_1,\dots,S_{m-1} \rp $ be a Markov chain with $q_j \lp S_k \rp $ being the transition probability of state $S_k$ to $S_j$ and the further technical condition that $\forall k,  q_0 \lp S_k \rp  > 0.$ This Markov chain is ergodic and therefore has a unique stationary distribution $\boldsymbol{\pi}:=(\pi(t_0), \pi(t_1),..., \pi(t_{m-1}))$.

The Markov Chains have  {\em rewards}.  That is, each state $S_k$ has a real-valued {\em reward} or {\em cost} $\cost (S_k).$
The  {\em gain} 
\cite{gallager2011discrete}  or (average steady state) cost of the Markov Chain is  
$\Cost(\Sbold)=\sum_{k\in [m]} \cost(S_k) \cdot \pi_k(\Sbold).$

For intuition, consider the AIVF code ${\bf t}$ as a $m=|S|-1$ state Markov chain, with tree $t_i$ being state $S_i.$ Set $\cost (S_i)$ to be $-E[L_{t_i}(S)]$. Then, the cost of the Markov chain is just $-E[L_{AIVF}(\bold{t}, S)]$. For technical reasons, we need to assume that the cost of each state is non-negative. But in our case, this is not an issue since as mentioned earlier, the average length of each code word is at most $D$. We can therefore add the constant $D$ to all costs, i.e. we can redefine $\Cost(S_i)$ to be $-E[L_{t_i}(S)] + D$, which is non-negative.

Now suppose that, for each $k,$ there is a collection  $\States_k$ of all "permissible"  possible states of type $k$. Their Cartesian product 
$\States = \bigtimes_{k=0}^{m-1} \States_k$ is the set of all permissible Markov Chains.
In the AIVF example,  $\States_k
=\mathcal{T}_{k}^{(D)}$ is the set of every possible $t_k$ trees and $\States=\mathcal{T}^{(D)}$ is the set of all possible AIVF codes.

The  {\em Minimum Cost Markov Chain (MCMC)  problem} is to find a Markov chain in $\States$ with  minimum cost among all Markov Chains in $\States.$   In the AIVF setting, defining $\Cost(S_i)=-E[L_{t_i}(S)] + D$,
 is exactly the problem of finding the optimal AIVF code.




The MCMC problem was originally  motivated by AIFV-$m$ (Almost Instantaneous Fixed to Variable) coding
\cite{iterative_3,iterative_4,aifv_mr,aifv_m,dp_2,iterative_m,iterative_2}
which encodes source characters with $m$ different types of  coding trees and switching between them using a complicated set of rules.
The algorithm developed for finding a minimum cost AIFV-$m$ code was an exponential time iterative one.
As noted in \cite{iterative_4}, this algorithm actually solved the more general
MCMC Problem.  Iwata and  Yamamoto \cite{AIVF-iterative}  then showed that the problem of finding the maximum cost AIVF code was a special case of the MCMC problem (as illustrated above) and used that iterative algorithm to solve the MCMC problem.

Very recently,   \cite{golinaifv-m}  described how to transform the generic MCMC problem into a linear programming (LP) problem. In applications, the associated polytope is described by an exponential number of constraints but the linear program could still be solved in polynomial time using the Ellipsoid method  \cite{ellipsoid} if a local optimization problem on underlying Markov chain states could be solved in polynomial time.

The remainder of the paper describes the conditions required to use the techniques in \cite{golinaifv-m} and then shows how the AIVF problem satisfies those conditions.  This leads to a polynomial time algorithm for finding the maximum cost AIVF code.

\section{More Definitions} 

The next set of definitions are from \cite{golinaifv-m} and stated in terms of Markov Chains and states. They map states to hyperplanes in $\R^m$ and define a polytope.



\begin{definition} \label{hp def} Let $k \in [m].$ In what follows, $\x$ will always satisfy 
        $\x = \left[x_1,x_2\dots,x_{m-1}\right]^T \in \mathbb{R}^{m-1}$.

    \begin{itemize}
        \item Define $f_k: \mathbb{R}^{m-1} \times \States_k \rightarrow \mathbb{R}$ as follows:
        \[
            f_0 \lp \x,S_0 \rp  = \cost \lp S_0 \rp  + \sum_{j=1}^{m-1}q_j \lp S_0 \rp x_j
        \]
        \[
            \forall k > 0: f_k \lp \x,S_k \rp  = \cost \lp S_k \rp  + \sum_{j=1}^{m-1}q_j \lp S_k \rp x_j - x_k
        \]
        \item Define $\g_k : \mathbb{R}^{m-1} \rightarrow \mathbb{R}$ as follows:
        \[
            g_k \lp \x \rp  = \min_{S_k \in \States_k} f_k \lp \x,S_k \rp,\ \  
            S_k \lp \x \rp  = \arg\min_{S_k \in \States_k} f_k \lp \x,S_k \rp,
        \]
        \[
            \Sbold \lp \x \rp  = \bigl( S_0 \lp \x \rp , \ldots, S_{m-1} \lp \x \rp  \bigr).
        \]
      
        $g_k(x)$ is the \textit{lower-envelope of type $k$}.
        \item Define $h: \mathbb{R}^{m-1} \rightarrow \mathbb{R}$ as 
        $
            h \lp \x \rp  = \min_{k} g_k \lp \x \rp.
        $
        
The {\em Markov Chain Polytope} corresponding to $\States$ is 
$$ \MCP = \left\{(\x,y) \in \mathbb{R}^{m} \mid  0 \le  y \le h(\x)\right\}.$$
The {\em height} of $\MCP$ is ${\rm height}(\MCP)=\max_{(\x,y)\in \MCP} .$
%
 %
    \end{itemize}

\end{definition}

The main observation from \cite{golinaifv-m} is the following, which
permitted transforming the MCMC problem into the linear programming one of finding  a highest point on $\MCP.$
\begin{proposition}[Lemma 3.1  and Corollary 3.4 in ~\cite{golinaifv-m}] \label{prop: big paper}
Let $ {\bf S} =(S_0,\ldots,S_{m-1})\in \mbbS$  and 
     $f_k \lp \x,S_i \rp$, $k \in [m],$ its  associated  hyperplanes. Then these $m$ hyperplanes intersect  at a unique point 
    $(\x,y) \in \mathbb{R}^{m}$.
    Such a point  $(\x,y)$  will be called the multi-typed intersection corresponding to $S$.  In addition, $ \Cost({\bf S}) =y \ge {\rm height}(\MCP).$

    Furthermore, there exists a  $ {\bf S}^* =(S^*_0,\ldots,S^*_{m-1})\in \mbbS$
    such that its associated multi-typed intersection point  $(\x^*,y^*) $
    satisfies $ y^* = {\rm height}(\MCP).$  This $\bf S$ is a solution to the MCMC problem and also satisfies $\forall k \in m, f_k(\x^*,S^*_k)={\rm height}(\MCP)$.
 %
 %
\end{proposition}


To deal with issues arising from possible transient states in the Markov chain, \cite{golinaifv-m} also needed 
\begin{definition}[$P$-restricted search spaces] \ 
\label{def:P_RestrictedX}\\ Set 
$\calP=\{ P\subseteq [m] \mid 0 \in P\}.$

For every state $S,$ define $P(S) = \{j \in [m] \,:\, q_j(S) >0\},$ the set of all states to which $S$ can transition.  Note that  $P(S) \in \calP.$

Now  fix  $k \in[m]$ and $P \in \calP.$
 Define $\mbbS_{k|P}$
to be the subset of all states $S_k\in\mbbS_k$ that only transition to states in $P,$
$$\mbbS_{k|P} = \Bigl\{ S_k \in\mbbS_{k}  \mid   P(S_k) \subseteq P\Bigr\}
\quad{\mbox and}\quad
\mbbS_{|P} = \bigtimes_{i=0}^{m-1} \mbbS_{k|P}.$$

        Further define  $g_{k|P} : \R^{m-1} \to \R$  and
        $S_{k|P} : \R^{m-1} \to\mbbS_{k|P}$
        \begin{eqnarray*}
        g_{k|P}(\x) &=& \min_{S_k \in \mbbS_{k|P}} f_{k}(\x, S_k).\\\
         S_{k|P} (\x) &=& \arg \min_{S_k \in \mbbS_{k|P}} f_{k}(\x, S_k),
        \end{eqnarray*}
        and
        $\Sbold_{|P} (\x)=  
        \left(S_{0|P}(\x),S_{1|P}(\x) \ldots,S_{m-1|P} (\x)
        \right).
        $


\end{definition}

The final definition needed from \cite{golinaifv-m} is:
\begin{definition}
Let ${\bf R} = \bigtimes_{k \in [m]} [l_k,r_k]$ be a hyperrectangle in $\R^{m-1}.$  Define  
$t'_{\States}({\bf R})  $ to be the maximum time required to calculate $ \Sbold_{|P}(\x)$ for any $\x\in {\bf R} $ and any $P \in \mathcal{P}.$ 
%
%
\end{definition}
While this might look mysterious, in the AIVF case calculating $\Sbold_{\x|[m]}(\x)$
is {\em exactly} the local optimization problem defined in 
Section \ref{sec: dp for aivf}.  Calculating $\Sbold_{\x|P}$ is the generalized local optimization problem described there.  Thus, for AIVF, 
$t'_{\States}({\bf R}) =O(D^2 |S|)$ independent of  $\bf R.$ For context, \cite{golinaifv-m} uses this to construct a {\em separation oracle} for $\MCP,$ something which is needed in the Ellipsoid method.

After introducing the framework, we can now state the major result from 
\cite {golinaifv-m} that we will use.
\begin{lemma}{\rm \cite[Lemma 4.6]{golinaifv-m}} 
\label{lem:newsol}
Given $\States,$ let $b'$ be the maximum number of bits required to write any transition probability  $q_i(S)$ or cost $\cost(S)$ of a permissible state $S$.

Furthermore,  assume  some known  hyper-rectangle  
${\bf R}$ with the property that there exists   $(\x^*,y^*) \in \mathbb{H}$ satisfying $\x^* \in {\bf R}$ and 
$y^* = {\rm height}(\MCP).$

%

Then the minimum-cost Markov chain problem can be solved in time polynomially bounded by  $m$, $b'$  and $t'_S({\bf R}).$

\end{lemma}

\section{The Main Result}  \label{sec: statements}
We already saw how the AIVF coding problem can be recast as an MCMC problem.
To solve it in polynomial time it only remains to show that the conditions of  Lemma \ref{lem:newsol} apply.
Note that we have already seen that, independent of the specific parameters of $\bf R,$ using the generalized dynamic programming algorithm,  $t'_S({\bf R})= O(D^2 |S|)$.

Now suppose that  
each $p(a_i),$
$i\in \mathcal{I}_{|S|}$
 can be written using at most
$b$ bits.  This implies that 
$p(a_i) = P_i 2^{-b}$
where  $P_i$ is an integer satisfying $0 \le P_i \le 2^b.$
It is then straightforward to bound the number of bits needed in encoding the problem.
More specifically, we  prove (proof in the appendix) that
\begin{lemma} \label{lem:bits}
In the MCMC formulation of the AIVF coding problem let $S_k \in\States_k$ be an arbitrary state.  Then 
each  $q_j(t_i)$ and  $\cost(S_i)$ can be encoded using at most $O(D b)$ bits. 
\end{lemma}

It only remains to exhibit a bounded  hyper-rectangle $\bf R$ that contains
 $(\x^*,y^*) \in \mathbb{H}$ where  
$y^* = {\rm height}(\MCP).$




First, we prove a Lemma (proof in the appendix) that introduces a bounding box containing the highest point of $\MCP$ if all probabilities and the cost function $\cost$ are bounded.

\begin{lemma} \label{lem: bounding lemma}
Set $\langle n\rangle =\{1,\ldots n\}.$
Let 
    $\States_i$ be as defined in the MCMC problem and further assume that  $\cost(S_i)$ is bounded by a positive real number $N$ for all $S_k \in \States_i.$ Moreover assume $q_0(S_k) \geq \beta$ for all $ k \in \langle m-1 \rangle$ and $S_k \in \States_k$. If $C \geq \beta^{-1} N$, then for all $i \in [m-1]$ the following statement holds:

    If $\x = [x_1,x_2,\dots,x_{m-1}]^T \in \R^{m-1}$ satisfies  $x_j \in [-C, C]$ for all $j \in \langle m-1\rangle \setminus \{i\}$, we have:

    \begin{enumerate}
        \item If $x_i > C$ then $f_i \lp \x , S_i \rp < 0;$
        \item If $x_i < -C$ then $f_i \lp \x , S_i \rp > N.$
    \end{enumerate}
\end{lemma}

\begin{remark}\label{rmk: bound beta}
    For the AIVF-m Markov  and its associated Markov Chain, we have  
    $0 \le \Cost(S_i)=-E[L_{t_i}(S)] + D\le D$. So in Lemma \ref{lem: bounding lemma}  we may assume that $N =D.$

Further note that every tree $t_i$ has to contain a leaf node, which has a codeword associated with it. This immediately implies that 
 $q_0(S_i) > 0$ for for all $S_i$,  Lemma \ref{lem:bits}  then implies that 
$q_0(S_i) \ge 2^{-c b D}$ for some constant $c >0.$ So, in Lemma \ref{lem: bounding lemma},  we may assume that $\beta=2^{-c b D}$.
\end{remark}



We need one final Theorem.
\begin{theorem} \label{thm:final}
Let ${\bf S}^*$ be the solution to the MCMC problem and $Q = (\x^*, h^*)$ its associated multitype intersection point as given by 
Proposition \ref{prop: big paper}.
Then 
    $\x^*\in [-\beta^{-1}N, \beta^{-1}N]^{m-1}$.
\end{theorem}

\begin{proof}
First note that by Proposition \ref{prop: big paper}, 
$\Cost({\bf S}^*) = y^* = {\rm height}(\MCP).$ 
Since  $0\le \cost(S^*_i)\le N$ for all $i \in [m]$ and 
$\Cost({\bf S}^*)$ is some  convex combination of the $\cost(S^*_i)$, we derive
$0\le \Cost({\bf S}^*)\le N$.

Let $\x^* = [x^*_1,x^*_2,\ldots, x^*_{m-1}]^T$ and set  $C = \max \{ |x^*_1| , \ldots, |x^*_{m-1}| \}$. If $\x^* \notin [-\beta^{-1}N, \beta^{-1}N]^{m-1}$, then there exists index $i \in \langle m-1 \rangle$ s.t. 
 \[
    |x^*_i| =  C > \beta^{-1}N
 \]
 but 
  $x^*_j \in [-C, C]$ for all $j \in \langle m-1 \rangle \setminus \{i\}$. Thus, by Lemma \ref{lem: bounding lemma}, one of the following two cases must occur.
 \begin{enumerate}
     \item $f_i \lp \x^*, S^*_i \rp < 0$
     \item $f_i \lp \x^*, S^*_i \rp > N$
 \end{enumerate}
 But, again by  Proposition \ref{prop: big paper},  $f_i \lp \x^*, S^*_i \rp= \Cost({\bf S}^*),$ leading to a contradiction.
\end{proof}


    



We may therefore apply Lemma \ref{lem:newsol} with $b' = O(Db)$ 
(Lemma \ref{lem:bits}) and ${\bf R} = [-D 2^{-bcD},D 2^{-bcD}]^{m-1}$
(Remark \ref{rmk: bound beta}  and  Theorem \ref{thm:final})
deriving that a maximum cost AIVF code can be constructed in time
polynomial in $m=|S|-1$, $b'=O(Db)$ and $t'_S({\bf R})= O(D^2 |S|)$, i.e.,  polynomial in $|S|$, $D$,  $b.$

\section*{Acknowledgment}
Work of all authors partially supported by RGC CERG Grant 16212021.

\newpage

\ 
\medskip

\IEEEtriggeratref{4}



\bibliography{refrence}

\newpage
\

\appendix
\subsection{Dynamic Programming for AIVF:}


This section provides more detail about the dynamic programming algorithms for AIFV that were discussed in Section \ref{sec: dp for aivf}.

In order to find optimal AIVF codes we are required to solve a local optimization problem. The dynamic programming algorithm we use is the same algorithm proposed by Iwata and Yamamoto in \cite{AIVF-iterative} which is an extension of the algorithm proposed by Dubé and Haddad ~\cite{dube-haddad} but has been rewritten using change of variables to permit it to fit into the MCMC format required by Lemma \ref{lem:newsol}.

Since  \cite{AIVF-iterative}didn't   provide full details of their modification and we require  a generalization of it, we provide the full algorithm with analysis here.

\begin{algorithm} 
\caption{Fills in the $\OPT(i:d)$ table and constructs the $t^{(i)}_d$ satisfying $\Cost\left(t^{(d)}_i\right)=\OPT(i:d)$ }\label{alg:dp}
\begin{algorithmic}[1]
\Require $D \geq 2$,  $\{\alpha_i|i \in \mathcal{I}_{|S|-1}\}$ and $\x=(x_1,\ldots,x_{|S|-2})$.\\
\textbf{Ensure:} $t_i^{(D)}$ and $\OPT(i:D)$ satisfy $t_i^{(D)}:=\argmax\limits_{s_i^{(D)} \in \mathcal{T}_{i}^{(D)}} \Cost\left(s^{(D)}_i\right)$ and $\OPT(i:D)=\Cost\left(t^{(D)}_i\right)$ respectively for $i\in \mathcal{I}_{|S|-1}$.
\algrenewcommand\algorithmicrequire{\textbf{Initialization:}}
\Require $t_i^{(1)} \gets $Root and $\OPT(i:1) \gets 0 $ for $ i\in{I}_{|S|-1}\ $.
\For{$d=2$ to $D$ }
\For{$i=0$ to $|S|-3$ }
\State  $\OPT(i:d) \gets  \max\limits_{1 \le l < d-1}
            \Bigl( \alpha_i + \alpha_i \cdot \OPT(0:l)+(1- \alpha_i) \cdot \OPT(i+1:d-l) \Bigr)$ \label{line:opt update}
\State  \hspace*{0.56in}$l' \gets  \argmax\limits_{1 \le l < d-1}
            \Bigl( \alpha_i + \alpha_i \cdot \OPT(0:l)+(1- \alpha_i) \cdot \OPT(i+1:d-l) \Bigr)$         \label{line:maximum l}
\State   $t^{(d)}_i \gets \Tie_i\left( t^{(l')}_0, t^{(d-l')}_{i+1} \right)$ \label{line:t_i update}           
\EndFor
\State  $\OPT(|S|-2:d) \gets  \max\limits_{1 \le l < d-1}
            \Bigl( 1 + \alpha_{|S|-2} \cdot \OPT(0:l)+(1- \alpha_{|S|-2}) \cdot \OPT(0:d-l) \Bigr)$
\State  \hspace*{0.93in}$l' \gets  \argmax\limits_{1 \le l < d-1}
            \Bigl( 1 + \alpha_{|S|-2} \cdot \OPT(0:l)+(1- \alpha_{|S|-2}) \cdot \OPT(0:d-l) \Bigr)$     
\State   $t^{(d)}_{|S|-2} \gets \Tie_i\left( t^{(l')}_0, t^{(d-l')}_{0} \right)$           \label{line:second t_i update} 
\EndFor \\
\Return $\left\{\OPT(i:D), t^{(D)}_i \mid i \in\mathcal{I}_{|S|-1}\right\}$
\end{algorithmic}
\end{algorithm}

\begin{figure}
    \centering
    \tikzset{every picture/.style={line width=0.75pt}} 

\begin{tikzpicture}[x=0.75pt,y=0.75pt,yscale=-1,xscale=1]

\draw   (113.07,74.05) -- (155.62,137.81) -- (113.07,137.81) -- cycle ;
\draw   (75.02,88.27) -- (93.04,136.37) -- (57,136.37) -- cycle ;
\draw  [line width=3] [line join = round][line cap = round] (75.02,89.23) .. controls (75.02,89.23) and (75.02,89.23) .. (75.02,89.23) ;
\draw  [line width=3] [line join = round][line cap = round] (113.07,74.05) .. controls (113.07,74.05) and (113.07,74.05) .. (113.07,74.05) ;
\draw    (75.02,88.27) -- (113.07,74.05) ;
\draw   (253.72,92.28) -- (278,144.88) -- (229.44,144.88) -- cycle ;
\draw   (195.44,92.28) -- (219.72,144.88) -- (171.16,144.88) -- cycle ;
\draw  [line width=3] [line join = round][line cap = round] (195.73,91.44) .. controls (195.73,91.44) and (195.73,91.44) .. (195.73,91.44) ;
\draw  [line width=3] [line join = round][line cap = round] (254,92.28) .. controls (254,92.28) and (254,92.28) .. (254,92.28) ;
\draw  [line width=3] [line join = round][line cap = round] (226.58,71.4) .. controls (226.58,71.4) and (226.58,71.4) .. (226.58,71.4) ;
\draw    (195.44,92.28) -- (226.58,71.4) ;
\draw    (226.58,71.4) -- (253.72,92.28) ;

\draw (85.04,62.89) node [anchor=north west][inner sep=0.75pt]    {$a_{i}$};
\draw (64.27,112.53) node [anchor=north west][inner sep=0.75pt]    {$s_{0}^{( l)}$};
\draw (116.33,114.13) node [anchor=north west][inner sep=0.75pt]    {$s_{i+1}^{( r)}$};
\draw (53.81,154.67) node [anchor=north west][inner sep=0.75pt]    {$Tie_{i}\left( s_{0}^{( l)} ,\ s_{i+1}^{( r)}\right)$};
\draw (182.66,116.66) node [anchor=north west][inner sep=0.75pt]    {$s_{0}^{( l)}$};
\draw (243.36,117.5) node [anchor=north west][inner sep=0.75pt]    {$s_{0}^{( r)}$};
\draw (172.01,64.47) node [anchor=north west][inner sep=0.75pt]    {$a_{\mid S\mid-2}$};
\draw (235.29,61.64) node [anchor=north west][inner sep=0.75pt]    {$a_{\mid S \mid-1}$};
\draw (172.86,156.25) node [anchor=north west][inner sep=0.75pt]    {$Tie_{\mid S \mid-2}\left( s_{0}^{( l)} ,\ s_{0}^{( r)}\right)$};

\end{tikzpicture}
    \caption{On the left-hand we can see an Illustration of how the trees in $\mathcal{T}_i$ can be decomposed for $i\in\mathcal{I}_{|S|-2}$ and on the right-hand side we can see how a the trees in $\mathcal{T}_{|S|-2}$ can be decomposed}
    \label{fig:tree dissection}
\end{figure}
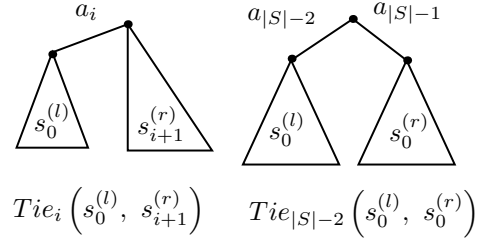

As before,  $\mathcal{T}_{i}$ denotes the set of trees $t_i$ satisfying (p1)-(p3), and $\mathcal{T}_{i}^{(D)}$ is the set of trees that are in $\mathcal{T}_{i}$ that have exactly $D$ codewords. Consider $\x=(x_1,\ldots,x_{|S|-2}) \in \R^{|S|-2}$. We define 
\begin{equation}
    \Cost\left(t_i,\x \right)=E[L_{t_i}(S)] + \sum_{j=1}^{|S|-2}q_j(t_i)\cdot x_j
\end{equation}
for $t_i\in\mathcal{T}_{i}$.

In order to find optimal AIVF codes we are required to solve a local optimization problem that involves finding the tree $t_0 \in \mathcal{T}_{0}^{(D)}$ that maximizes $E[L_{t_0}(S)] + \sum_{j=1}^{|S|-2}q_j(t_0)\cdot x_j$ and the tree $t_i \in \mathcal{T}_{i}^{(D)}$ that maximizes $E[L_{t_i}(S)] + \sum_{j=1}^{|S|-2}q_j(t_i)\cdot x_j-x_i$ for every $i>0$. Since the $x_i$s are fixed this is equivalent to finding the tree $t_i \in \mathcal{T}_{i}^{(D)}$ that maximizes $\Cost\left(t_i,\x\right)$ for each $i\in\mathcal{I}_{|S|-1}$. We now present a few simple observations that will help us develop a straightforward approach to find these trees.

For $s_{0}^{(l)} \in \mathcal{T}_{0}^{(l)}$, $s_{i+1}^{(r)} \in \mathcal{T}_{i+1}^{(r)}$ and $i\in\mathcal{I}_{|S|-2}$ we define $\Tie_{i}\Bigl(s_0^{(l)},s_{i+1}^{(r)}\Bigr)$ to be the unique tree in $\mathcal{T}_{i}$ that is created by taking $s_{i+1}^{(r)}$ and adding a new edge labeled $a_i$ to its root that connects to $s_0^{(l)}$. Consequently this new tree has $l+r$ codewords and is therefore in $\mathcal{T}_{i}^{(l+r)}$. Note that the converse is true, that is for any $s_{i}^{(D)}\in\mathcal{T}_{i}^{(D)}$ satisfying $D\geq2$ and $i\in\mathcal{I}_{|S|-2}$ there exists $s_{0}^{(l)} \in \mathcal{T}_{0}^{(l)}$ and $s_{i+1}^{(r)} \in \mathcal{T}_{i+1}^{(r)}$ such that $l,r\geq1$, $l+r=D$ and $s_{i}^{(D)}:=\Tie_{i}\Bigl(s_0^{(l)},s_{i+1}^{(r)}\Bigr)$. We can see an illustration of this on the left-hand side of Figure \ref{fig:tree dissection}.

For $s_{0}^{(l)} \in \mathcal{T}_{0}^{(l)}$ and $s_{0}^{(r)} \in \mathcal{T}_{0}^{(r)}$ we define $\Tie_{|S|-2}\Bigl(s_0^{(l)},s_{0}^{(r)}\Bigr)$ to be the unique tree in $\mathcal{T}_{|S|-2}$ that is created by taking a new node as the root having 2 edges labeled $a_{|S|-2}$ and $a_{|S|-1}$ connecting to $s_{0}^{(l)}$ and $s_0^{(r)}$ respectively. Consequently this new tree has $l+r$ codewords and is therefore in $\mathcal{T}_{|S|-2}^{(l+r)}$. Note that similar to the previous case, the converse is true, that is for any $s_{|S|-2}^{(D)}\in\mathcal{T}_{|S|-2}^{(D)}$ satisfying $D\geq2$ there exists $s_{0}^{(l)} \in \mathcal{T}_{0}^{(l)}$ and $s_{0}^{(r)} \in \mathcal{T}_{0}^{(r)}$ such that $l,r\geq1$, $l+r=D$ and $s_{|S|-2}^{(D)}:=\Tie_{|S|-2}\Bigl(s_0^{(l)},s_{0}^{(r)}\Bigr)$. We can see an illustration of this on the right-hand side of figure \ref{fig:tree dissection}.

Next consider $s_{i}^{(D)}\in\mathcal{T}_{i}^{(D)}$ satisfying $D\geq2$ and $i\in\mathcal{I}_{|S|-2}$ and  the unique trees $s_{0}^{(l)} \in \mathcal{T}_{0}^{(l)}$ and $s_{i+1}^{(r)} \in \mathcal{T}_{i+1}^{(r)}$ satisfying $s_{i}^{(D)}:=\Tie_{i}\Bigl(s_0^{(l)},s_{i+1}^{(r)}\Bigr)$. 
Consider encoding an input string using $s_{i}^{(D)}$. The probability that the first character of our input is $a_i$ is equal to $\frac{p(a_i)}{\sum_{j=i}^{|S|-1}p(a_j)}$. Denote $\alpha_{i}=\frac{p(a_i)}{\sum_{j=i}^{|S|-1}p(a_j)}$. This means that the probability that we will parse our next codeword using $s_{0}^{(l)}$ and $s_{i+1}^{(r)}$ is $\alpha_i$ and $1-\alpha_i$ respectively. Using this observation we can say:
\begin{equation}
    \forall{j\in\mathcal{I}_{|S|-1}}:\space q_j\Bigl(s_i^{(D)}\Bigr)=\alpha_i 
 \cdot q_j\Bigl(s_0^{(l)}\Bigr)+(1-\alpha_i) \cdot q_j\Bigl(s_{i+1}^{(r)}\Bigr)
\end{equation}
\begin{equation}
    E[L_{s_{i}^{(D)}}(S)]=\alpha_i \cdot \Bigl(E[L_{s_0^{(l)}}(S)]+1\Bigr)+(1-\alpha_i) \cdot E[L_{s_{i+1}^{(r)}}(S)]
\end{equation}
Hence
\begin{eqnarray} \label{eq:type i cost decomposition} \nonumber
&\Cost(s_i^{(D)},\x)=E[L_{s_i^{(D)}}(S)]+ \sum_{j=1}^{|S|-2}q_j\Bigl(s_i^{(D)}\Bigr)\cdot x_j\\ \nonumber
&=\alpha_i \cdot \Bigl(E[L_{s_0^{(l)}}(S)]+1\Bigr) \nonumber \\
&+(1-\alpha_i) \cdot E[L_{s_{i+1}^{(r)}}(S)] \nonumber \\
&+\sum_{j=1}^{|S|-2}\Biggl(\alpha_i \cdot q_j\Bigl(s_0^{(l)}\Bigr)+(1-\alpha_i) \cdot q_j\Bigl(s_{i+1}^{(r)}\Bigr)\Biggr)\cdot x_j \nonumber \\ 
&=\alpha_i+\alpha_{i} \cdot \Biggl(E[L_{s_0^{(l)}}(S)]+\sum_{j=1}^{|S|-1}q_j\Bigl(s_0^{(l)}\Bigr)\cdot x_j\Biggr) \nonumber \\
&+(1-\alpha_i) \cdot \Biggl(E[L_{s_{i+1}^{(r)}}(S)]+q_j\Bigl(s_{i+1}^{(r)}\Bigr)\cdot x_j\Biggr) \nonumber \\ 
&=\alpha_i+\alpha_{i}\cdot \Cost\Bigl(s_0^{(l)},\x\Bigr)+(1-\alpha_i)\cdot \Cost\Bigl(s_{i+1}^{(r)},\x\Bigr) 
\end{eqnarray}
Now consider $s_{|S|-2}^{(D)}\in\mathcal{T}_{|S|-2}^{(D)}$ satisfying $D\geq2$ and consider the unique trees $s_{0}^{(l)} \in \mathcal{T}_{0}^{(l)}$ and $s_{0}^{(r)} \in \mathcal{T}_{0}^{(r)}$ satisfying $s_{|S|-2}^{(D)}:=\Tie_{|S|-2}\Bigl(s_0^{(l)},s_{0}^{(r)}\Bigr)$. Denote $\alpha_{|S|-2}=\frac{p(a_{|S|-2})}{p(a_{|S|-2})+p(a_{|S|-1})}$. Similar to before, we can say:
\begin{eqnarray}
    \nonumber
    \Cost\Bigl(s_{|S|-2}^{(D)},\x\Bigr)&=&E[L_{s_{|S|-2}^{(D)}}(S)]+ \sum_{j=1}^{|S|-2}q_j\Bigl(s_{|S|-2}^{(D)}\Bigr)\cdot x_j\\ \nonumber
    &=&\alpha_i \cdot \Bigl(E[L_{s_0^{(l)}}(S)]+1\Bigr) \nonumber\\
    &+&(1-\alpha_i) \cdot \Bigl(E[L_{s_{0}^{(r)}}(S)]+1\Bigr) \nonumber \\
    &+&\sum_{j=1}^{|S|-2}\Biggl(\alpha_i q_j\Bigl(s_0^{(l)}\Bigr) \nonumber\\
    &+&(1-\alpha_i) \cdot q_j\Bigl(s_{0}^{(r)}\Bigr)\Biggr)\cdot x_j \nonumber\\
    &=&1+\alpha_{|S|-2}\cdot \Cost\Bigl(s_0^{(l)},\x\Bigr)\nonumber \\
    &+&(1-\alpha_{|S|-2})\cdot \Cost\Bigl(s_{0}^{(r)},\x\Bigr) \nonumber
\end{eqnarray}
These observations lead to a divide and conquer approach that will help us use dynamic programming to solve this problem.

\begin{theorem} \label{theorem:Tunstal DP}
The output $\left\{\OPT(i:D), t^{(D)}_i \mid i \in\mathcal{I}_{|S|-1}\right\}$ of Algorithm \ref{alg:dp} satisfies $t_i^{(D)}=$ $\argmax\limits_{s_{i}\in \mathcal{T}_{i}^{(D)}} \Cost(s_i,\x)$ and $\OPT(i:D)=\Cost\Bigl(t_i^{(D)},\x\Bigr)$ for $i\in\mathcal{I}_{|S|-1}$.
\end{theorem}

\begin{proof}
By induction on that $d$ we show that $t_{i}^{(d)}$ and $\OPT(i:d)$ derived from Algorithm \ref{alg:dp} satisfy 
$$t_i^{(d)}:=\argmax\limits_{t_{i}\in \mathcal{T}_{i}^{(d)}} \Cost(t_i,\x)$$ and $$\OPT(i:d)=\Cost\Bigl(t_i^{(d)},\x\Bigr)$$ for $i\in\mathcal{I}_{|S|-1}$

For the base case we have $t_i^{(1)}:=$Root and $\OPT(i:1)=0$ for $i\in\mathcal{I}_{|S|-1}$ where Root is a tree with only one node. For the inductive step, consider the statement true for all the values less than $d$ and consider the tree $s_{i}^{(d)} \in \mathcal{T}_{i}^{(d)}$ where $i\in\mathcal{I}_{|S|-2}$ and $s_i^{(d)}=$ $\argmax\limits_{s_{i}\in \mathcal{T}_{i}^{(d)}} \Cost(s_{i},\x)$.

As mentioned before, we know there exists unique trees $s_{0}^{(l)} \in \mathcal{T}_{0}^{(l)}$ and $s_{i+1}^{(d-l)} \in \mathcal{T}_{i+1}^{(d-l)}$ such that $1 \leq l \leq d-1$ and $s_{i}^{(d)}:=\Tie_{i}\Bigl(s_0^{(l)},s_{i+1}^{(d-l)}\Bigr)$. Now considering equation \ref{eq:type i cost decomposition}, in order for $s_{i}^{(d)}$ to maximize $\Cost\Bigl(s_{i}^{(d)},\x\Bigr)$, $s_{0}^{(l)} \in \mathcal{T}_{0}^{(l)}$ and $s_{i+1}^{(d-l)} \in \mathcal{T}_{i+1}^{(d-l)}$ must maximize $\Cost\Bigl(s_{0}^{(l)},\x\Bigr)$ and $\Cost\Bigl(s_{i+1}^{(d-l)},\x\Bigr)$ respectively.

In other words $s_{0}^{(l)}:=t_{0}^{(l)}$ and $s_{i+1}^{(d-l)}:=t_{i+1}^{(d-l)}$. This means that in order to find $s_i^{(d)}$ we can run  through all possibilities of $l$ and see in which case $\Cost\Bigl(s_i^{(d)},\x\Bigr)$ is maximized. This  is what's happening in line \ref{line:maximum l} and \ref{line:t_i update} of Algorithm \ref{alg:dp}.

We can also see that in line \ref{line:opt update}, $\OPT(i:d)$ is being calculated accordingly.
In a similar fashion we can show that the tree  $s_{|S|-2}^{(d)}$ satisfying $$s_{|S|-2}^{(d)}:=\argmax\limits_{s_{|S|-2}\in \mathcal{T}_{|S|-2}^{(d)}} \Cost(s_{|S|-2},\x)$$ can be found by running  through all possibilities of $l$ satisfying $1\leq l\leq d-1$ and observing for which value it maximizes $$\Cost\Biggl(\Tie_{|S|-2}\Bigl(t_0^{(l)},t_{0}^{(d-l)}\Bigr),\x\Biggr)$$ and calculating $\OPT(|S|-2:d)$ accordingly. This completes our inductive step and therefore proves the theorem.
\end{proof}


\begin{proposition} \label{prop: DP running time appendix}
    Algorithm \ref{alg:dp} can find the final trees in $O(D^2\cdot |S|)$ running time.
\end{proposition}
\begin{proof}
    Consider the following modification; We work through Algorithm \ref{alg:dp} but instead of doing the Tie operations at line \ref{line:t_i update} and \ref{line:second t_i update} we just save the corresponding value of $l^{\prime}$ for each $d$ and $i$ in a separate table. The reason we do this is because we don't actually need to create every tree and we can just use this new table to recursively build the final trees after the loops are completed. This doesn't change the time complexity but it does decrease the space complexity.
    
    Call this new table $E$. This means that $E(i:d)$ tells us how many codewords are in the left subtree of $t_{i}^{(d)}$ when decomposed according to Figure \ref{fig:tree dissection}. Using this information we can recursively create the left and right subtrees in order to find $t_{i}^{(d)}$. Creating the final trees with this procedure has time complexity of $O(|S|\cdot D)$ because at each step we are adding at least one edge and then recursively finding the left and right subtrees. Since each step takes $O(1)$ time and our trees have at most $2D$ nodes, we create each tree in $O(D)$ time and find all the final trees in $O(|S|\cdot D)$ which doesn't affect the time of filling in the tables resulting in an $O(D^2\cdot |S|)$ overall run time.
\end{proof}

\textbf{Generalized Dynamic Programming:} \newline
Consider $P\subseteq \mathcal{I}_{|S|-1}$ such that $0\in P$ and consider the problem of adding the constraint $q_j(t_{i}^{(d)})=0$ for every $j\notin P$ 
and $i\in \mathcal{I}_{|S|-1}$. This means we want to find the trees having the optimal costs defined previously, with the added constraint that each tree has a transition probability of zero for trees having a type that is not in $P$. As discussed in the main body of the paper, we need this in order to be able to apply Lemma \ref{lem:newsol}

Modifying the previous dynamic program to handle this generalization is not difficult. Using the decomposition in figure \ref{fig:tree dissection}, we can see that in order for $s_{i}^{(d)}$ to have the optimal cost between trees that have the added constraint, its subtrees need to be optimal trees that satisfy the added constraint, which means that we can use almost the exact same dynamic programming formulation.

The only difference is that the initialization is going to be different. Consider $j\in \mathcal{I}_{|S|-2}$ and let $i\in \mathcal{I}_{|S|-2}$ be the smallest number satisfying $j < i$ and $i \in P$.  If no such $i$ exists then let $i=|S|-1$. We can see that in this case there are no type $j$ trees having less than $i-j$ codewords and there is exactly one type $j$ tree having $i-j$ codewords which is created by taking a root node having $i-j$ edges with $a_j,a_{j+1},..., a_{i-1}$ as their labels. Using this fact we can initialize the dynamic programming by first adding this information into the table. The rest of the dynamic programming will be almost exactly the same, except that we have to slightly modify it to avoid adding trees that are not permitted.

The running time for this generalized version will remain the same  $O(D^2\cdot |S|)$.

\subsection{Proofs missing from the main paper}
\textbf{Proof of Lemma \ref{lem:bits}}.
\begin{proof}
Since every node $w$ in a $t_i$ is the probability of at most $D$ $p(a_i)$ values,
each node's probability $p_{W}(w)$ can be written as $P' 2^{-bD}$ where 
$P'$ is an integer satisfying $0 \le P' \le 2^{bD}.$ 
Then, from Equation (\ref{eq:pwi}), $p_{W_i}(w)$ can also be written 
as $P'' 2^{-bD-1}$ where 
$P''$ is an integer satisfying $0 \le P'' \le 2^{bD+1}.$

According to equation  (\ref{eq:pci}), $p_{C_i}(w)$ is the sum of at most $D$ such probabilities. Thus, $p_{C_i}(w)$ can be represented by $\log_2(D) + bD+1 = O(b D)$ bits.

Since the nodes corresponding to dictionary words have disjoint probabilities,
the sum of the costs of any subset of those nodes is also a probability 
that can be written using  $O(b D)$ bits. Thus every 
$q_j(t_i)$ can be written using at most $O(b D)$  bits.

Recall from Equation (\ref{eq:apl}) that 
 the average length of code words in  a type $i$ tree $t_i$ is 
    $$
        E[L_{t_i}(S)] = \sum_{w \in \mathcal{D}(t_i)} p_{C_i}(w) \cdot l(w).
 $$

Moreover, $l(w)$ is upper-bounded by $D$. Therefore, each term in 
the summand can be written using at most $O(b D)$  bits.
Since $E[L_{t_i}(S)]$ is the sum of at most $D$ of these terms, $E[L_{t_i}(S)]$ can be represented by $O(bD)$ bits.

Recall that our transformation from the AIVF problem to the MCMC problem required setting $\cost(S_i)=D - E[L_{t_i}(S)].$ But this only adds at most another $O(\log_2 D)$ bits to the representation so we are done.
\end{proof}

\textbf{Proof of Lemma \ref{lem: bounding lemma}}
\begin{proof}
    We start by proving  (1). Note that
    \[
        f_i \lp \x , S_i \rp = \cost \lp S_i \rp  + \sum_{j=1}^{m-1}q_j \lp S_i \rp x_j - x_i
    \]
    Since for $j \in \langle m-1 \rangle \setminus \{i\}$ we have $x_j \in [-C, C]$, it is clear that  
    $$    f_i \lp \x , S_i \rp \leq \cost \lp S_i \rp + \sum_{j \in \langle m-1 \rangle \setminus \{i\}} q_j \lp S_i \rp C - \lp 1 - q_i(S_i) \rp x_i
   $$
    which can be rewritten as 
    \[
        f_i \lp \x , S_i \rp \leq \cost \lp S_i \rp + \lp 1 - q_0(S_i) - q_i(S_i) \rp C - \lp 1 - q_i(S_i) \rp x_i.
    \]
    We claim that
    \[
        N + \lp 1 - q_0(S_i) - q_i(S_i) \rp C - \lp 1 - q_i(S_i) \rp C \leq 0
    \]
    which is true if and only if
    \[
        q_0(S_i) C \geq N.
    \]
    Note that $q_0(S_k) \geq \beta$ and $C \geq \beta^{-1} N$; therefore, the above inequality is true and the claim is proved. Moreover, it is easy to further conclude that:
    \begin{align*}
            \cost \lp S_i \rp + \lp 1 - q_0(S_i) - q_i(S_i) \rp C - \lp 1 - q_i(S_i) \rp x_i < \\
            N + \lp 1 - q_0(S_i) - q_i(S_i) \rp C - \lp 1 - q_i(S_i) \rp C
    \end{align*}
So,  case (1) of the lemma is proved.
    
    For the second part, by a similar argument 
    we have
    \[
        f_i \lp \x , S_i \rp \geq \cost \lp S_i \rp -  \lp 1 - q_0(S_i) - q_i(S_i) \rp C - \lp 1 - q_i(S_i) \rp x_i
    \]
    Since $\cost \lp S_i \rp \geq 0$ we have
    \begin{align*}
        f_i \lp \x , S_i \rp &> -  \lp 1 - q_0(S_i) - q_i(S_i) \rp C + \lp 1 - q_i(S_i) \rp C \\
        &= q_0(S_i)C \geq N
    \end{align*}

    and cae (2) of the Lemma is proved.
\end{proof}

\end{document}